\documentclass[twocolumn,prb,showpacs]{revtex4}
\AtBeginDocument{
\heavyrulewidth=.08em
\lightrulewidth=.05em
\cmidrulewidth=.03em
\belowrulesep=.65ex
\belowbottomsep=0pt
\aboverulesep=.4ex
\abovetopsep=0pt
\cmidrulesep=\doublerulesep
\cmidrulekern=.5em
\defaultaddspace=.5em
}

\usepackage{amsmath,amssymb}
\usepackage{graphicx}
\usepackage[usenames,dvipsnames]{xcolor}
\usepackage{amsthm}
\usepackage{booktabs}
\usepackage{hyperref}
\newtheorem{myprop}{Observation}

\hypersetup{colorlinks=true,citecolor=Orange,linkcolor=Red,urlcolor=Black}

\begin{document}

\title{Symmetry Conserving Purification of Quantum States within the Density Matrix Renormalization Group}

\author{A. Nocera  and G. Alvarez}
\affiliation{Computer Science \& Mathematics %
Division and Center for Nanophase Materials Sciences, Oak Ridge National Laboratory, %
 \mbox{Oak Ridge, Tennessee 37831}, USA}

\begin{abstract}
The density matrix renormalization group (DMRG) algorithm was originally designed
to efficiently compute the zero temperature or ground-state properties of one 
dimensional strongly correlated quantum systems. The development of the algorithm
at finite temperature has been a topic of much interest, because of the usefulness
of thermodynamics quantities in understanding the physics of condensed matter
systems, and because of the increased complexity associated with efficiently computing 
temperature-dependent properties. The ancilla method is a DMRG technique that enables
the computation of these thermodynamic quantities. In this paper, we review the
ancilla method, and improve its performance by working on reduced Hilbert spaces and
using canonical approaches. We furthermore
explore its applicability beyond spins systems to t-J and Hubbard models.
\end{abstract}


\maketitle
\newpage

\section{Introduction}
The behavior of low dimensional strongly correlated electron systems at finite
temperature has became a topic of great interest amongst the condensed matter
physics community. Tuning the temperature can lead to interesting phenomena
which cannot be fully understood in terms of ground state physics,
such as the quantum critical regime,\cite{Sachdev2007} the
transition to a spin incoherent Luttinger liquid,\cite{Fiete2007} or the sudden
emergence of a single spinon dispersion in XXZ-like spin-chain
materials.\cite{James2009,Nagler1983,Villain1975} On the experimental front, recent
advances\cite{Zhu2005,Lake2005,Xu2007,Zvyagin2012}
have enabled very precise measurements of both static and dynamical correlation
functions. These experimental results are challenging the available theoretical and computational 
techniques,\cite{Grossj2009,Sandvik1991,Sandvik2010}
increasing demand to develop efficient and accurate numerical
methods for the investigation and eventual prediction of the thermodynamic quantities.

In one dimension, the density matrix renormalization group~\cite{White1992,White1993,Schollwock2011} (DMRG)
is the most powerful method to calculate ground state properties of strongly
correlated systems. The DMRG has been
successfully extended to treat real time dynamics~\cite{Feiguin2005,Kollath2004,Vidal2004,Schmitteckert2004}
and systems at finite temperature. These extensions include transfer-matrix 
DMRG,\cite{Sirker2005} minimally 
entangled typical thermal states or METTS,\cite{White2009,Stoudenmire2010,Bonnes2014} 
and the purification
 scheme~\cite{Feiguin2005,Barthel2009,Feiguin2010,Verstraete2004,Zwolak2004,Sirker2005,Karrasch2013,Tiegel2014}
to be used in this paper.

METTS was recently introduced by White in ref.~\onlinecite{White2009},
where he shows that these so-called minimally entangled typical thermal states
efficiently represent the thermal properties of the system of interest.
METTS has been applied to fermionic systems by one of us.\cite{Alvarez2013}
Ref.~\onlinecite{Binder2014} compares the purification scheme to the METTS method,
concluding that METTS is, in general, computationally less efficient than purification.
In particular, it is argued,\cite{Binder2014} that the additional statistical error
source introduced by METTS sampling increases computational costs, especially
at high temperatures. More recently, ref.~\onlinecite{Brugnuolo2015} develops
symmetry conserving METTS, improving the efficiency of the method.
Perhaps METTS and ancilla should be regarded as complementary.
Given these recent developments in DMRG techniques for temperature dependence, the present paper
revisits the ancilla method, with an emphasis on improving its performance, and
exploring its applicability beyond spins systems to t-J and Hubbard models.

The ancilla method was originally introduced within a thermofield 
formalism by Umezawa\cite{Umezawa1982}, and later in DMRG\cite{Zwolak2004,Feiguin2005}.
It is designed to calculate the finite temperature quantum average of an observable,
 by definition a trace over all the eigenstates of the physical chain,
and not accessible with standard DMRG.
To this aim, the above quantity is expressed 
as a quantum average over a pure state--- a thermal vacuum---in an enlarged 
Hilbert space.
The Hilbert space is enlarged by adding an ancilla site to each site of the
physical lattice, thus obtaining a two-leg ladder geometry. The
ancilla's degrees of freedom are added in such a way that the reduced
density matrix of the physical sites reproduces the thermal density matrix.

The purification scheme starts with an infinite temperature state, given
by a product of maximally entangled states on each rung of the ladder. For spin
chains,\cite{Feiguin2005} each ancilla can be chosen to have ``opposite'' quantum
numbers to those of the corresponding physical site. This choice of ancillas is
equivalent to applying a time-reversal transformation in the case of spin chains,
or, more generally, a particle-hole transformation on the ancillas, and has its
profound motivation in giving the correct Kubo-Martin-Schwinger 
relations\cite{Suzuki1985,Suzuki1986,Takahashi1975,Umezawa1982}
at finite temperature. The previously described choice has the additional
advantage of reducing the dimension of the Hilbert space of the total
system composed of physical and ancilla chains.
For example, in the case of a Heisenberg model in the absence of a magnetic field, 
one is able to work in a subspace of the enlarged Hilbert space of physical
plus ancilla chains with 
$z$ component of the total spin equal to zero. After the preparation of the 
initial state, a finite temperature state 
is then obtained by evolving in imaginary time with the Hamiltonian of the physical chain.

In this paper we show that the above approach, which we define as 
\emph{grand canonical} for the physical chain, is not the most efficient for the
DMRG numerical simulations: the preparation of the initial infinite temperature
state and the subsequent imaginary time evolution can be furthermore
restricted, without any additional truncation of the Hilbert space, to a subspace 
where the $z$ component 
of the total spin  of the physical chain is \emph{also} conserved. 
This subspace has clearly smaller dimension than the one previously mentioned,
the one that conserves \emph{only} the $z$ component of the total (physical
plus ancilla) system.

How can we impose the restriction outlined above? Or, more generally, how can we
work in the \emph{canonical ensemble} for the physical chain?
To answer this question we will show how to engineer an ``entangler'' Hamiltonian
that conserves a given set of quantum numbers of the physical chain,
and generates a maximally entangled infinite temperature state that is different
from the state used in the grand canonical approach. 
Introduced by Feiguin and Fiete in ref.~\onlinecite{Feiguin2010} for a t-J chain, it
turns out that the canonical ``entangler'' Hamiltonian is non-local for the total
system, but involves interactions between the rungs of the total system
at all possible distances. \emph{In the thermodynamic limit,}
the canonical scheme gives the same results as the grand canonical one. 
Yet the performance of the DMRG simulation is remarkably improved, by up
to one order of magnitude with respect to the grand canonical approach.

Section~\ref{sec:heisenberg} introduces the grand canonical and the canonical
purification schemes for the Heisenberg model. 
Section~\ref{sec:tj} then extends the treatment to the case of the t-J model, and section~\ref{sec:hubbard}
 to the Hubbard model.
The  Heisenberg model has only  spin degrees of freedom, but the t-J and Hubbard
models have also charge, leading to  
different ``entangler'' Hamiltonians. For each model, we first consider a local
\emph{entangler,} where neither the total number of electrons, nor the total number 
of spins of the physical chain are conserved; they are conserved only in
average at finite temperature. This is the \emph{grand canonical ensemble}. It needs a 
(temperature-dependent) chemical potential in order to keep constant the \emph{average} 
number of electrons in the physical chain. A magnetic field is likewise needed
to keep constant the 
\emph{average} $z$ component of the total spin (as in the case of a spin chain). 
We then consider a non local \emph{entangler}
where the number electrons is conserved but not the $z$-component of the
spin. Even in this
case, we show that the z-component of the spin for the electrons in
the physical chain is conserved in average during the temperature evolution; 
its value is zero in the absence of a magnetic field.
Finally, we write down an entangler such that \emph{both} the
number of electrons \emph{and} the z-component of the spin of the physical
chain are conserved.
 
\section{The ancilla method for finite temperature DMRG}

For convenience, we here summarize the ancilla approach for finite temperature
DMRG; refs.~\onlinecite{Feiguin2005,re:feiguin0813} have more detail. At $\beta=0$, 
that is, at infinite temperature, a maximally entangled
state $|\psi(\beta=0)\rangle$ between the physical sites of the chain and their
ancillas is initially produced. 
A pure state in the enlarged system at finite temperature is then calculated by
evolving $|\psi(\beta=0)\rangle$ in imaginary time with the Hamiltonian of 
the physical sites, $|\psi(\beta)\rangle = e^{-{\beta H/2}}|\psi(\beta=0)\rangle$. 
Given a generic observable $O$ of the physical chain, the thermodynamic average can be
calculated\cite{Suzuki1985,Suzuki1986,Takahashi1975,Umezawa1982,Feiguin2005} 
in the enlarged space using 
the standard zero temperature expression
\begin{equation}
\langle O \rangle = \frac{\langle \psi(\beta) |O| 
\psi(\beta)\rangle}{\langle \psi(\beta)|\psi(\beta)\rangle}, 
\end{equation}
where the norm $Z(\beta)=\langle \psi(\beta)|\psi(\beta)\rangle$ represents the
partition function at temperature $T\equiv 1/\beta$. 

We present two different purification schemes for the initial infinite
temperature state. The first scheme is referred to as ``grand 
canonical,'' where one conserves only quantum numbers for the enlarged system
given by the physical and the ancilla sites. 
In the second scheme, one conserves quantum numbers \emph{not just} globally,
that is, not just in the enlarged system, 
but in the physical and ancilla chains separately as well.

In the next sections we will often discuss intensive energies of the models considered.
The definition of intensive energy is 
\begin{equation}
\langle E \rangle/L = \frac{1}{L}\frac{\langle \psi(\beta) |H| \psi(\beta)\rangle}
{\langle \psi(\beta)|\psi(\beta)\rangle}, 
\end{equation} 
where $H$ is the Hamiltonian of the model under consideration and $L$ is the 
system size.

\subsection{Heisenberg model}\label{sec:heisenberg}
\subsubsection{Grand-canonical purification scheme}
Let us start considering the case of a spin chain described by the Heisenberg model
\begin{equation}
H_{Heis}=J\sum_{i=0}^{L-2}\vec{S}_{i}\cdot\vec{S}_{i+1},
\end{equation} 
with $\vec{S}=(S^{x},S^{y},S^{z})$, and where the Hilbert space of a single site is 
two-dimensional,
having only two possible states, $|\mathord{\uparrow}\rangle$ and 
$|\mathord{\downarrow}\rangle$. 
Recall briefly the purification scheme adopted in ref.~\onlinecite{Feiguin2005},
which in our classification turns out to have a grand-canonical character. 
For simplicity, first consider the case of two spins, accompanied by their
ancilla sites. The infinite temperature state is
\begin{equation}\label{State2spins}
|\psi_\text{2 spins}(\beta=0)\rangle = \frac12(|\uparrow \downarrow\rangle + 
|\downarrow \uparrow\rangle)
 \otimes (|\uparrow \downarrow\rangle + |\downarrow \uparrow\rangle),
\end{equation}
where $\otimes$ is the direct product of states on the two \emph{composite} sites, each
given by a pair consisting of a physical site and its ancilla. 
In this paper, \emph{the first entry of the ket vector refers to the state of the 
physical site and the second entry to the ancilla site.} For example, ket 
$|\mathord{\uparrow\downarrow}\rangle$, has a spin up in the physical site and a
spin down in the corresponding ancilla site. 
In the state Eq.~(\ref{State2spins}), each ancilla has ``opposite'' quantum
numbers with respect to those of the physical site.
As long as the mapping between the states of a physical site and its ancilla
is one-to-one, then any choice for the mapping of quantum numbers of the ancillas
yields the correct thermodynamics for the physical chain. 

The choice outlined above 
(which is equivalent to applying a time-reversal transformation\cite{Feiguin2005})
has the advantage of reducing the dimension
of the Hilbert space of the total enlarged system. By construction, the state
Eq.~(\ref{State2spins}) lives in the subspace
with spin $S^{z}_{\text{tot}}\equiv\sum_{i} ( S^{z}_{i} + S^{z}_{a(i)})=0$  of the total
physical-plus-ancilla Hilbert space, where 
we have indicated with $a(i)$ the ancilla site corresponding
to the physical site $i$.
For $L$ sites, Eq.~(\ref{State2spins}) generalizes to a so-called product of 
\emph{local} maximally-entangled states
\begin{equation}\label{MaxEntangled}
|\psi(\beta=0)\rangle = \frac{1}{\sqrt\mathcal{N}}\bigotimes_{i=0}^{L-1} 
\sum_{\sigma=\uparrow,\downarrow}|\sigma \bar\sigma\rangle,
\end{equation}
where $\mathcal{N}$ is a normalization constant, 
$\bar\uparrow=\downarrow$ and $\bar\downarrow=\uparrow$. The above state
is an eigenstate of the $z$-component of the total spin 
$S^{z}_{\rm tot}$, but \emph{it is not} an eigenstate of the total spin of either
the physical $S^{z}_{\text{ph.}}\equiv\sum_{i} S^{z}_{i}$
or ancilla sites $S^{z}_{\text{an.}}\equiv\sum_{i} S^{z}_{a(i)}$.

The above observation has important consequences for 
the efficiency of the finite temperature evolution. Indeed, 
within the purification scheme outlined above, 
one can study the thermodynamics of a generic spin chain Hamiltonian in 
the presence of an external magnetic field. 
The average magnetization of the system can be tuned by changing the 
strength of the field at finite temperature. 
This possibility stems from the symmetry property of the initial infinite
temperature state, which conserves
neither
$S^{z}_{\text{ph.}}$ nor $S^{z}_{\text{an.}}$.

A state like Eq.~(\ref{MaxEntangled}) is exponentially large, and would have to
be truncated before time evolving it. Otherwise, one would have to work with a vector of
 size $2^L.$
For the grand canonical scheme, one could build it by growing it slowly
inside a traditional DMRG, and
truncating it along the way. But due to the DMRG transformations,
the book keeping of such a state would be unfeasible to perform. 
It is then clear that with traditional DMRG the entangler
is needed for an efficient representation of the state.
With an MPS approach, the infinite temperature state in Eq.~(\ref{MaxEntangled}) has bond dimension 2, and is not entangled in a global sense, because it is just a product state of singlets on each rung.
How can we then generate state Eq.~(\ref{MaxEntangled}) as a starting point for the
temperature evolution? As proposed by Feiguin and Fiete~in
ref.~\onlinecite{Feiguin2010}, we find useful the notion of ``entangler''
Hamiltonian: a Hamiltonian
having the state Eq.~(\ref{MaxEntangled}) as its ground state. 
By diagonalizing a $4\times4$ matrix, where $4$ is the dimension of the composite 
physical plus ancilla site, the entangler for the Heisenberg model in the 
grand-canonical scheme can be written as
\begin{equation}\label{HeisenbergGC}
H^{spin}_{\text{GC}} = -\sum_{i=0}^{L-1} S^{+}_{i} S^{-}_{a(i)} + {\rm h.~c..}
\end{equation}

We address the results coming from this entangler in the next section,
where we compare the grand canonical and canonical purification schemes.

\subsubsection{Canonical purification scheme}

Let us consider the case of the Heisenberg model in the absence of
a magnetic field. In this case, one would build a 
maximally entangled state with the property that  $S^{z}_{tot}|\phi\rangle=0$,
 $S^{z}_{\text{ph.}}|\phi\rangle=0$, $S^{z}_{\text{an.}}|\phi\rangle=0$, 
such that one conserves the spin of the physical and ancilla 
chains \emph{separately.} 

For the simple case of two spins, a maximally entangled state with the
above characteristics is 
\begin{equation}\label{State2spinsC}
|\phi_\text{2 spins}(\beta=0)\rangle = \frac{1}{\sqrt2}
\big(|\uparrow \downarrow\rangle \otimes |\downarrow \uparrow\rangle + 
|\downarrow \uparrow\rangle \otimes |\uparrow \downarrow\rangle\big).
\end{equation}
For $L$ sites
\begin{equation}\label{MaxEntangledC}
|\psi(\beta=0)\rangle_{C} = \frac{1}{\sqrt{\mathcal{N}'}}P_{(S^{z}_{\text{ph.}}=0)}
\Bigg[\bigotimes_{i=0}^{L-1} \sum_{\sigma=\uparrow,\downarrow}|\sigma \bar\sigma\rangle\Bigg],
\end{equation}
where $\mathcal{N}'$
is a normalization constant, $P_{(S^{z}_{\text{ph.}}=0)}$ is the projector operator such that the 
$z$-component of the \emph{total} spin of the physical (ancilla) chain is conserved
and equal to zero: $S^{z}_{\text{ph.}}|\psi(\beta=0)\rangle_{C}=S^{z}_{\text{an.}}|\psi(\beta=0)\rangle_{C}=0$.
\begin{figure}
\centering
\includegraphics[]{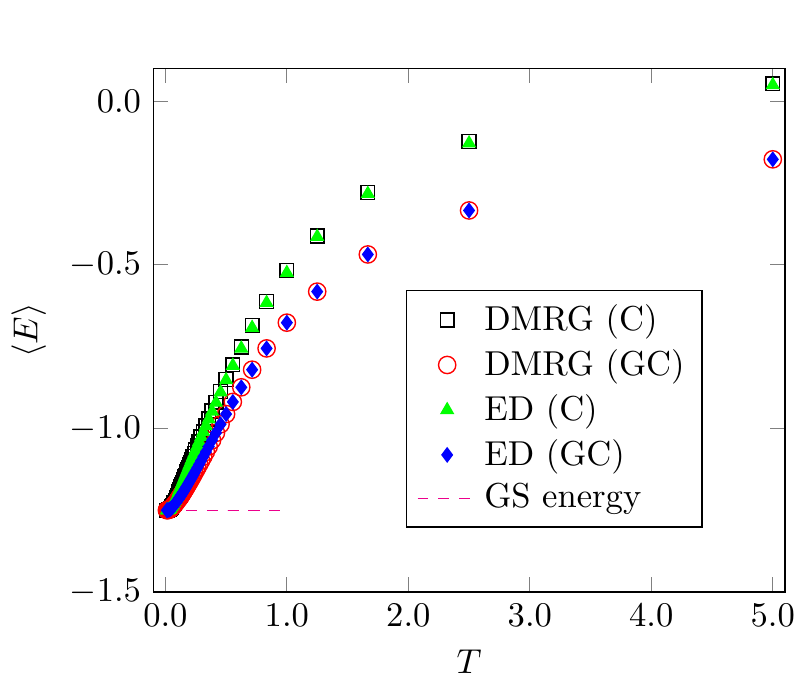}
\caption{(Color online) Intensive energies for a Heisenberg chain with $J=-1$ and
length $L=6$ sites for the canonical and grand canonical approach, and comparing 
DMRG and ED, as indicated.} \label{fig1}
\end{figure}

\begin{figure}
\centering
\includegraphics[]{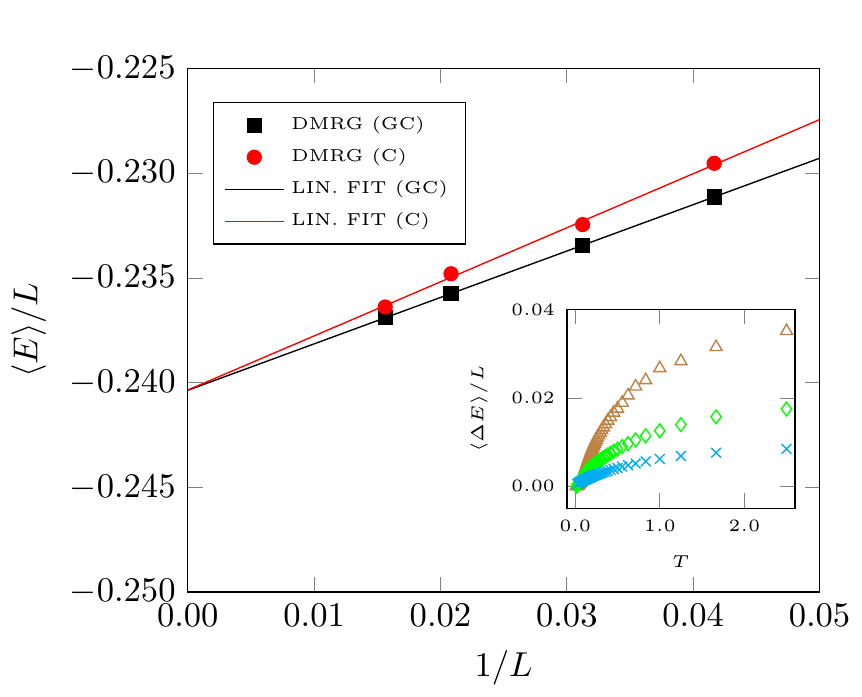}
\caption{(Color online)
Main panel: average intensive energies for the Heisenberg model ($J=-1$) at temperature
$T=0.1$ as a function of the system size comparing the canonical and the 
grand canonical approaches (we consider up to $L=64$ sites). 
The Bethe ansatz yields an energy equal to $-0.240219$ in the thermodynamic limit. Inset: difference 
between the intensive energies obtained in the two approaches, 
$\langle\Delta E\rangle=\langle E_{C}\rangle-\langle E_{GC}\rangle$, as function
of the temperature, for different system sizes: triangles (brown) indicate $L=6$,
diamonds (green) indicate $L=12$, and crosses (cyan) indicate $L=24$.} \label{fig2}
\end{figure}

For an $L$ site chain, the subspace containing 
the maximally entangled state in the canonical approach has dimension $C^{L}_{L/2}$
($C^{k}_{l}$ gives the number of possible combinations of $l$ 
objects on $k$ places), a number clearly smaller than the analogous subspace 
dimension $C^{2L}_{L}$ of the grand canonical case. As will be seen below,
imposing the conservation of $S^{z}_{\text{ph.}}$ 
during the finite temperature evolution 
in the enlarged Hilbert
space substantially reduces the computational effort of the time dependent DMRG 
simulations.\cite{Mcculloch2007}
Conserving this symmetry can thus help improve DMRG 
implementations, because they usually
consider \emph{global} symmetries only.

How can we, in practice, generate the maximally entangled ``canonical'' state 
Eq.~(\ref{MaxEntangledC})? By finding the correct ``entangler'' Hamiltonian,
the one that gives Eq.~(\ref{MaxEntangledC}) as its ground state.
For the Heisenberg model, this ``entangler'' Hamiltonian is
\begin{equation}\label{HeisenbergC}
H^{spin}_{\text{C}} = -\sum_{i\neq j} \Gamma^{\dag}_{i} \Gamma_{j} + {\rm h.~c.},
\end{equation}
where $\Gamma^{\dag}_{i}=S^{+}_{i} S^{-}_{a(i)}$.
Notice that the above Hamiltonian has long-range interactions: each 
rung interacts with all the others. Section \ref{sec:entanglement} discusses the 
entanglement implications caused by long-range interactions in the DMRG calculations.
A proof that the state (\ref{MaxEntangledC}) is the 
ground state of the Hamiltonian (\ref{HeisenbergC}) is provided in appendix~\ref{App:AppendixB}.

In matrix product operator notation,\cite{Schollwock2011}
 $H^{spin}_{\text{C}}=\hat{W}^{[0]}\hat{W}^{[1]}\cdots\hat{W}^{[L-1]}+$h.c.,
where 
\begin{equation}
\hat{W}^{[i]}=\left[\begin{tabular}{lll}
$I$ & $0$ & $\Gamma^\dagger_i$\\
$0$ & $I$ & $0$\\
$0$ & $\Gamma_i$ & $I$\\
\end{tabular}\right],
\end{equation}
for $0<i<L-1$, $\hat{W}^{[0]}=[I,0,\Gamma^\dagger_0]$, and
\begin{equation}
\hat{W}^{[L-1]}=\left[\begin{tabular}{l}
$0$\\
$I$\\
$\Gamma_{L-1}$\\
\end{tabular}\right].
\end{equation}

Fig.~\ref{fig1} shows the intensive energy defined as
as a function of the temperature for a ferromagnetic Heisenberg chain with $L=6$
 sites, and where we assume $|J|=1$ as unit of energy. Fig.~\ref{fig1} compares both
  purification schemes. 
It also compares DMRG with exact diagonalization. 
The figure shows a perfect numerical agreement between DMRG and ED, but 
substantially different results between the 
grand canonical and canonical approaches; the intensive energies calculated within 
the canonical approach are systematically bigger than those calculated in the 
grand canonical scheme. Fig.~\ref{fig1} also shows that
the results obtained with different approaches coincide for $\beta\rightarrow\infty$, 
that is, at zero temperature. 
In appendix~\ref{App:AppendixA}, we prove that for a generic interacting quantum system 
the canonical (when chosen in a symmetry sector containing the ground state)
and grand canonical average energies differ in general, but coincide at 
zero temperature, and in the thermodynamic limit (where they coincide at any temperature).

The DMRG implementation used throughout this paper is discussed in 
the supplemental material which can be found at \url{https://web.ornl.gov/~gz1/papers/55/}.
The algorithm implemented for the
time-dependent part of the DMRG is based on the Krylov space decomposition method, as explained
in ref.~\onlinecite{re:manmana05} and references therein.
An implementation can be found in ref.~\onlinecite{Alvarez2011}. In the time-evolution calculations,
600--800 states per block are kept. The results are well converged with this number of states in the 
entire interval of temperature investigated ($\beta_{min}=0$, $\beta_{max}/2=40$, with $\Delta\beta/2=0.1$).

Fig.~\ref{fig2} substantiates another analogy of the results obtained in the two
purification schemes with those obtained in different ensembles. 
The main plot of fig.~\ref{fig2} shows the average
energies obtained in the canonical and grand canonical approaches as a 
function of the system size at low temperature,
$T=0.1$. The results are different for small system size, but they
converge to the exact Bethe ansatz solution when extrapolated to the thermodynamic limit. 
The inset shows the difference between the average
energies in the grand canonical and canonical approaches, which decreases at \emph{every} temperature 
with increasing system size. On the other end, all the curves 
converge to the same 
value for $\beta\rightarrow\infty$, as already observed in fig.~\ref{fig1}. 
In the supplemental material, which can be found at \url{https://web.ornl.gov/~gz1/papers/55/}, 
we have proved that this is valid for a generic quantum system when the thermodynamic limit is taken.

\subsection{t-J chains}\label{sec:tj}

\subsubsection{Grand canonical purification scheme}

In this section, we apply the theory discussed previously
to the case of a $t-J$ model described by the standard Hamiltonian
\begin{eqnarray}\label{tJ}
H_{tJ} &=&-t\sum_{i=0,\sigma}^{L-2} \big( c_{i,\sigma}^{\dag}c_{i+1,\sigma} + h.c.\big) 
\nonumber\\
&+& J \sum_{i=0}^{L-2} \big(\vec{S}_{i}\cdot\vec{S}_{i+1}- N_{i}N_{i+1}/4\big),
\end{eqnarray}
where $N_{i}=\sum_{\sigma=\uparrow,\downarrow}c_{i,\sigma}^{\dag}c_{i,\sigma}$. The model is characterized by a three-dimensional single-site Hilbert space with states empty $|0\rangle$, single occupied with spin up $|\mathord{\uparrow}\rangle$, and spin down $|\mathord{\downarrow}\rangle$.
We start with the grand-canonical purification scheme. Following the same 
procedure introduced in the spin chains' case, one can straightforwardly 
write down a product of local (along the rungs of the ladder) maximally 
entangled states:
\begin{equation}\label{MaxEntangledtJ}
|\psi(\beta=0)\rangle = \bigotimes_{i=0}^{L-1} \Big[|0,0\rangle+
\sum_{\sigma=\uparrow,\downarrow}|\sigma \bar\sigma\rangle\Big].
\end{equation}
We have verified (by diagonalizing a $9\times9$ matrix,
where 9 is the dimension of a composite physical-plus-ancilla site) that this state can be generated 
by calculating the ground state of the entangler Hamiltonian
\begin{equation}\label{EntanglertJGC}
H^{t-J}_{\text{GC}} = \sum_{i=0}^{L-1} \big(\sqrt{2}\Delta_{i}+S^{+}_{i} S^{-}_{a(i)}\big) +{\rm h.~c.},
\end{equation}
where $\Delta^{\dag}_{i}=(c^{\dag}_{i,\uparrow}c^{\dag}_{a(i),\downarrow}-
c^{\dag}_{i,\downarrow}c^{\dag}_{a(i),\uparrow})/\sqrt{2}$.
As observed in the case of the spin chains, even 
though Eq.~(\ref{MaxEntangledtJ}) conserves the total spin in the enlarged physical plus ancilla combined system, it 
\emph{does not conserve these quantities separately} for the physical 
chain or for the ancilla chain. In order to get the thermodynamic 
properties for the physical chain at finite temperature in the grand canonical 
ensemble, one must add a chemical potential term $\mu$ during the imaginary 
time evolution. To keep 
an average of $N$ electrons in the physical chain, one needs to solve the 
  equation 
\begin{equation}\label{Equatmu}
\langle \psi(\beta,\mu)|N_{\text{ph.}}|\psi(\beta,\mu)\rangle = N, 
\end{equation}
for each temperature $T=1/\beta$, 
where $|\psi(\beta,\mu)\rangle=e^{-{\beta/2} [H_{tJ}-\mu N_{\text{ph.}}]} 
|\psi(\beta=0)\rangle$, with $H_{tJ}$ being the standard $t-J$ Hamiltonian and 
$N_{\text{ph.}}\equiv\sum_{i} N_{i,\rm ph}$. For this reason, when  
studying the thermodynamic properties of the t-J model at fixed 
density and zero magnetic field, the grand canonical 
scheme outlined above has a clear disadvantage in terms of computational cost. 
Therefore, results in the grand canonical 
scheme are calculated only with exact diagonalization for a small system size.
We have verified that the average 
 spin of the physical chain is zero at any 
temperature, $\langle \psi(\beta,\mu)|S^{z}_{\text{ph.}}|\psi(\beta,\mu)\rangle=0$, 
with $\mu$ being a solution of Eq.~(\ref{Equatmu}).

\subsubsection{Canonical purification scheme}\label{subsubsec:tj}

In this section, we address the ``canonical'' purification scheme for the $t-J$ model. 
We present two different schemes which take into account the charge and the
spin symmetries of the physical chain.  We first review a canonical purification
method where one only employs the charge conservation in the physical chain.
This scheme has already been treated by Feiguin and Fiete in ref.~\onlinecite{Feiguin2010} 
in the context of a spin-incoherent Luttinger liquid. 
We consider a t-J chain with $L$ sites, and a even 
number of electrons $N$, such that $N_{\uparrow}=N_{\downarrow}=N/2$. Notice that in the total system given by physical and ancilla sites one has $2N$ electrons.
Up to a constant, we here provide the expression
\begin{equation}\label{MaxEntangledCtJ}
|\psi(\beta=0)\rangle_{\text{C1}} = P_{(N_{\text{ph.}}=N)}\Bigg[\bigotimes_{i=0}^{L-1} 
\big(|0,0\rangle+\sum_{\sigma=\uparrow,\downarrow}|\sigma \bar\sigma\rangle\big)\Bigg]
\end{equation}
for the ``canonical'' maximally 
entangled state in the $L$ sites case,
where $P_{(N_{\text{ph.}}=N)}$ is a projector operator (different from 
the projector that appears in Eq.~(\ref{MaxEntangledC})) such that the total 
number of electrons of the physical (ancilla) chain is conserved, 
$N_{\text{ph.}}|\psi(\beta=0)\rangle_{\text{C1}}=N_{\text{an.}}|\psi(\beta=0)\rangle_{\text{C1}}=N$.
We emphasize that the above state does \emph{not} conserve the the $z$-component 
of the total spin of the physical (ancilla) chain, but it has the property that, at any 
temperature, ${}_{\text{C1}}\langle \psi(\beta)|S^{z}_{\text{ph.}}|\psi(\beta)\rangle_{\text{C1}}=0$. 

We now prove that the ``canonical'' entangler Hamiltonian 
\begin{equation}\label{EntanglertJC1}
H^{t-J}_{\text{C1}} = -\sum_{i\neq j} \Delta^{\dag}_{i} \Delta_{j} + {\rm h.~c.}
\end{equation}
proposed in ref.~\onlinecite{Feiguin2010}
has the property that its ground-state is Eq.~(\ref{MaxEntangledCtJ}).
We begin by observing that the Hamiltonian (\ref{EntanglertJC1}) conserves 
the number of electrons $N$ in the physical and ancilla chains \emph{separately}, 
and assume that the combined physical plus ancilla system has $S_z=0$, so that 
the number of up and down electrons are equal in the
combined system. Let us divide the full Hilbert space basis into ``good states'' where
 all physical sites and their ancillas are correctly paired,
 and ``bad states'' where at least one physical site is not correctly paired to its
ancilla; these disjoint sets are then
\begin{gather}\label{goodbad}
\begin{aligned}
\mathcal{S}_G=&\{|\phi\rangle; \text{basis } |\phi\rangle\text{ with all physical and ancilla}\\
&\text{ sites correctly paired}\}\\
\mathcal{S}_B=&\{|\phi\rangle; \text{basis } |\phi\rangle\text{ with at least one physical site}\\
&\text{ with ancilla incorrectly paired}\}.
\end{aligned}
\end{gather}
Note that if $|\phi\rangle\in\mathcal{S}_G$ then $|\phi\rangle$ is a term in Eq.~(\ref{MaxEntangledCtJ}). 
We shall prove that (i) $\langle \phi'|H^{t-J}_{\text{C1}}|\phi\rangle=0$ if $|\phi'\rangle\in\mathcal{S}_G$
  and $|\phi\rangle\in\mathcal{S}_B$, 
 so that the Hamiltonian matrix $H^{t-J}_{\text{C1}}$ blocks into at least
two blocks: states in $\mathcal{S}_G$ and states in $\mathcal{S}_B$.
We shall furthermore prove (ii) that the ground state of $H^{t-J}_{\text{C1}}$ is in the block
  $\mathcal{S}_G$ as opposed to the block $\mathcal{S}_B$, 
 and (iii)
that the ground state of the block $\langle \phi'|H^{t-J}_{\text{C1}}|\phi\rangle$ for all
$|\phi\rangle,|\phi'\rangle\in\mathcal{S}_G$ is Eq.~(\ref{MaxEntangledCtJ}).

To prove (i) it suffices to think of Eq.~(\ref{EntanglertJC1}) as a 
tight-binding Hamiltonian of ``singlets''
of physical sites and ancillas built along the rungs, singlets that
are hopping along the ladder structure. Eq.~(\ref{EntanglertJC1}) cannot 
connect states in $\mathcal{S}_G$
with those in $\mathcal{S}_B$.

To prove (ii) we first note that if $|\phi\rangle\in\mathcal{S}_B$, then at 
least one site
has occupation $|\sigma,0\rangle$, $|0,\sigma\rangle$ or $|\sigma,\sigma\rangle$.
Electrons in sites like these cannot ``move''
by the action of  Eq.~(\ref{EntanglertJC1}). The subspace
$\mathcal{S}_B$ then blocks into even smaller subspaces, all of size smaller 
than the set $\mathcal{S}_G$.
Within these subspaces, we have hopping-like matrices 
yielding lowest energies larger than that of the block  $\mathcal{S}_G$.

To prove (iii), let us call $H$ the block of matrix Eq.~(\ref{EntanglertJC1}) in the
$\mathcal{S}_G$ subspace.
The Hamiltonian Eq.~(\ref{EntanglertJC1}) conserves the number of electrons in the
 physical system, but not its total spin, so that it can be written in a tridiagonal block form.
The diagonal blocks have rank   
 $C^L_{N_{\uparrow}}C^{L-N_{\uparrow}}_{N-N_{\uparrow}}$ each, and are
 characterized by the number of up electrons. The blocks in the first 
 diagonal above
(below) the main diagonal connect states with $N_{\uparrow}$ and $N_{\uparrow}+1$
($N_{\uparrow}-1$) electrons and are given by $m\times n$ rectangular matrices with 
$m = C^L_{N_{\uparrow}}C^{L-N_{\uparrow}}_{N-N_{\uparrow}}$ and 
$n = C^L_{N_{\uparrow}+1}C^{L-N_{\uparrow}-1}_{N-N_{\uparrow}-1}$ 
($n = C^L_{N_{\uparrow}-1}C^{L-N_{\uparrow}+1}_{N-N_{\uparrow}+1}$).

Now, the matrix elements in the blocks described above are either 0 or -1 for the diagonal blocks while for the off-diagonal ones are 0 and +1. 
Each row has exactly $2N(L-N)$ non-zero entries. 
So does each column.
Then its lowest eigenvalue (by taking into account the factor $\sqrt{2}$ in the
definition of the $\Delta$ operators) is
$-N(L-N)$, with eigenstate $(1,1,\cdots,1)$, that is,  
Eq.~(\ref{MaxEntangledCtJ}).

\begin{figure}
\centering
\includegraphics[]{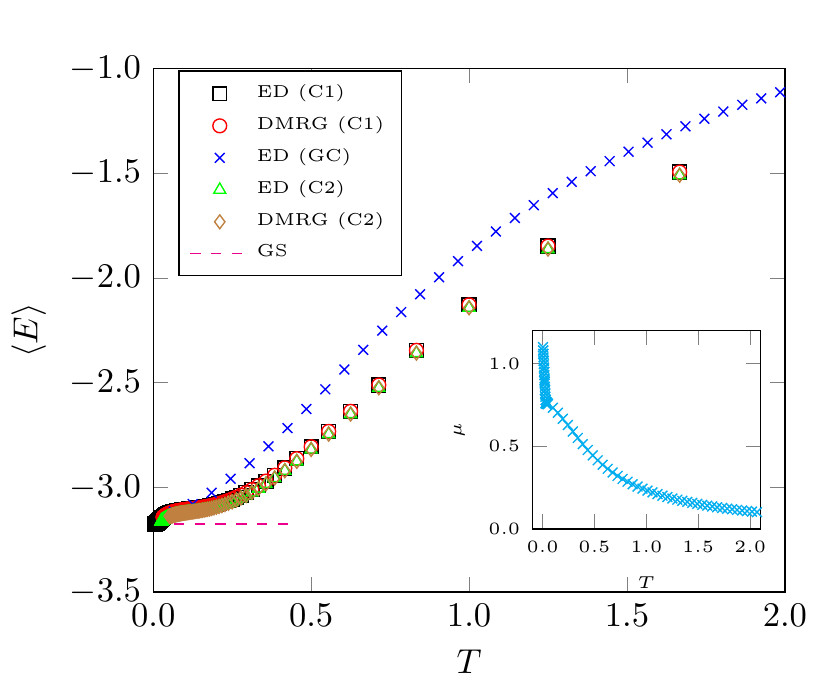}
\caption{(Color online) Intensive energies for a t-J chain of $L=6$ sites, $N=4$ 
electrons, with $J=0.1$,
for the canonical and grand canonical approaches ($\text{C1}$ and C2, see main text), 
and comparing DMRG against ED, as indicated.
\emph{Inset:} Chemical potential as a function of the temperature, calculated in 
the grand canonical approach
so that the average density is always $N=4$ electrons.} \label{fig3}
\end{figure}

Because the $t-J$ chain has charge and spin degrees of freedom, there is an
alternative and more efficient way of performing the canonical purification
scheme: Consider the maximally entangled infinite temperature
state 
\begin{equation}\label{MaxEntangledCtJ1}
|\psi(\beta=0)\rangle_{\text{C2}} = P_{\binom{N_{\text{ph.}}=N}{S^{z}_{\text{ph.}}=0}}
\Bigg[\bigotimes_{i=0}^{L-1} \big(|0,0\rangle+\sum_{\sigma=\uparrow,\downarrow}|\sigma \bar\sigma\rangle\big)\Bigg],
\end{equation}
where the projector operator is such that the total number of electrons
\emph{and} the total $z$-component of the physical (ancilla) chain are
conserved,
\begin{align}
N_{\text{ph.}}|\psi(\beta=0)\rangle_{\text{C2}}=&N|\psi(\beta=0)\rangle_{\text{C2}}\\
N_{\text{an.}}|\psi(\beta=0)\rangle_{\text{C2}}=&N|\psi(\beta=0)\rangle_{\text{C2}}\\
S^{z}_{\text{ph.}}|\psi(\beta=0)\rangle_{\text{C2}}=&S^{z}_{\text{an.}}
|\psi(\beta=0)\rangle_{\text{C2}}=\emptyset.
\end{align} 

With a procedure similar to the one just outlined, it is possible to show that an
``entangler'' Hamiltonian yielding Eq.~(\ref{MaxEntangledCtJ1}) as ground state is 
\begin{equation}\label{EntanglertJC2}
H^{t-J}_{\text{C2}} = -\sum_{\substack{i\neq j\\\sigma=\uparrow,\downarrow}} c^{\dag}_{i,\sigma}
 c^{\dag}_{a(i),\bar\sigma}c_{j,\bar\sigma}c_{a(j),\sigma} + {\rm h.~c.}.
\end{equation}
The ground state energy of the above Hamiltonian is equal to $-N(L-N)$, coinciding with
the result for the $\text{C1}$ Hamiltonian. 
Notice that, for an $L$ site chain, the subspace containing 
the maximally entangled state in the canonical approach $\text{C1}$ has dimension 
$D_{\text{C1}}=\sum_{N_{\uparrow}=0}^{N} C^{L}_{N_{\uparrow}}C^{L-N_{\uparrow}}_{N-N_{\uparrow}}$, 
which is evidently larger than the analogous subspace dimension $D_{\text{C2}}=C^{L}_{N/2}C^{L-N/2}_{N/2}$ 
of the case C2, where one has also $N_{\uparrow}=N_{\downarrow}=N/2$. 
If one imposes the conservation of $N_{\text{ph.}}$ \emph{and} $S^{z}_{\text{ph.}}$ 
during the finite temperature evolution in the enlarged Hilbert
space, this remarkably reduces the computational effort and increases the 
efficiency of the time dependent DMRG simulations. In our typical runs, when this purification 
scheme is adopted, the computational time needed for obtaining the thermodynamic properties 
is reduced by a factor of two with respect to the canonical purification C1.

Fig.~\ref{fig3} shows the average intensive energy as a function of the temperature 
for a t-J chain with $L=6$ sites, $N=4$ electrons, and $J=0.1$, 
($t=1$ is assumed as unit of energy) comparing canonical and grand canonical 
approaches. 
One observes a perfect numerical agreement between DMRG and ED for both canonical 
approaches C1 and C2. For the small system size
considered, the results for the average energies differ (by less than $1\%$) at temperature larger than zero,
 while converging to the ground state value in the limit of large $\beta$. 
We have verified that, even for a system size of $L=12$ sites, the numerical results obtained 
in the C1 and C2 approaches are coincident within the numerical precision of our runs. 

The results obtained in the grand canonical approach were calculated with ED by
performing a procedure similar to a standard Maxwell construction. We first
calculated the total electronic density as a function of the temperature, for
different chemical potential values $\mu$. By solving an equation similar to
Eq.~(\ref{Equatmu}), $\langle N_{tot}\rangle=4$, we have then extracted the
chemical potential curve at constant density as function of the temperature, which
is reported in the inset of fig.~\ref{fig3}. Finally, in the main plot of fig.~\ref{fig3},
the blue crosses represent the intensive energy as a function of the temperature. 
Even for small system sizes, DMRG calculations 
in the grand canonical approach are already computationally expensive. We can then conclude 
that the canonical purification C2 should be the computationally preferred 
choice for the study of the thermodynamic properties of the $t-J$ model. (Yet the disadvantages
of the canonical are discussed in section~\ref{sec:entanglement}.)

\subsection{Hubbard chains}\label{sec:hubbard}
\label{sub:Hubbard}

\subsubsection{Grand canonical purification scheme}

In this section, we apply the purification schemes introduced in the previous section 
to a one dimensional Hubbard model described by the standard Hamiltonian
\begin{equation}\label{Hubbard}
H_{Hub}=-t\sum_{i=0,\sigma}^{L-2} \big( c_{i,\sigma}^{\dag}c_{i+1,\sigma} + h.c.\big) 
+U \sum_{i=0}^{L-1} N_{i,\uparrow}N_{i,\downarrow}.
\end{equation}
The model is characterized by a local four-dimensional Hilbert space with states empty
 $|0\rangle$, single occupied $|\mathord{\uparrow}\rangle$, 
$|\mathord{\downarrow}\rangle$ and double occupied 
$|\mathord{\uparrow\downarrow}\rangle$.

We start by 
considering the grand-canonical purification scheme. Following the same procedure 
used for the $t-J$ case, one can straightforwardly write (up to a constant)
\begin{equation}\label{MaxEntangledHub}
|\psi\rangle = \bigotimes_{i=0}^{L-1} \Big[|0,0\rangle+
|\uparrow\downarrow,\uparrow\downarrow\rangle+
\sum_{\sigma=\uparrow,\downarrow}|\sigma \bar\sigma\rangle\Big]
\end{equation}
for the infinite 
temperature state in the grand canonical approach.
The only difference with respect to the state Eq.~(\ref{MaxEntangledtJ}) is that 
 the double occupied state on a physical site is ``mapped'' to the 
same state on its ancilla. By diagonalizing a $16\times16$ matrix, 
where 16 is the dimension of a composite physical-plus-ancilla site, the state above 
can be generated by calculating the ground state of the entangler Hamiltonian 
\begin{equation}\label{EntanglerHubGC}
H^{\text{Hubbard}}_{\text{GC}} = \sum_{i=0,\sigma=\uparrow,\downarrow}^{L-1}\big( 
c_{i,\sigma}c_{a(i),\bar\sigma}P^{\sigma}_{i} + {\rm h.~c.}\big),
\end{equation}
where $P^{\sigma}_{i}=|1-N_{i,\bar\sigma}-N_{a(i),\sigma}|$.

State Eq.~(\ref{MaxEntangledHub}) conserves  the total spin 
 in the enlarged chain, but does not do so 
for the physical chain. A chemical potential 
 $\mu$ must therefore be added
 in the Hamiltonian used during 
the temperature evolution, and \emph{mutatis mutandis}, the same procedure outlined in the 
previous section for the $t-J$ case must be followed.
The average total spin of the physical chain is zero at any 
temperature, that is, $\langle \psi(\beta,\mu)|S^{z}_{tot}|\psi(\beta,\mu)\rangle=0$, 
with $\mu$ being a solution of an equation similar to Eq.~(\ref{Equatmu}).
Because of the computational cost of the procedure, the results in
the grand canonical scheme are presented using only exact diagonalization on a small system size.

\subsubsection{Canonical purification scheme}

We now address the ``canonical'' purification schemes for the Hubbard model.
We consider a chain with $L$ sites, filling $N/L$, and a 
total even number of electrons $N$, such that $N_{\uparrow}=N_{\downarrow}=N/2$.

Up to a constant, the ``canonical'' maximally entangled state 
of type C1 at infinite temperature is
\begin{equation}\label{MaxEntangledC1Hub}
|\psi\rangle_{\text{C1}} = P_{(N_{\text{ph.}}=N)}\Bigg[\bigotimes_{i=0}^{L-1} \big(|0,0\rangle+|\uparrow\downarrow,
\uparrow\downarrow\rangle+\sum_{\sigma=\uparrow,\downarrow}|\sigma \bar\sigma\rangle\big)\Bigg],
\end{equation}
where $P_{(N_{\text{ph.}}=N)}$ is a projector operator such that the total 
number of electrons of the physical (ancilla) chain is conserved, 
\begin{equation}
N_{\text{ph.}}|\psi_{\beta=0}\rangle_{\text{C1}}=N_{\text{an.}}|\psi_{\beta=0}\rangle_{\text{C1}}=
N|\psi_{\beta=0}\rangle_{\text{C1}}.
\end{equation}
Even though this state conserves charge in the physical chain, 
it does \emph{not} conserve the $z$-component of the spin of the physical 
chain. Yet, ${}_{\text{C1}}\langle \psi(\beta)|S^{z}_{\text{ph.}}|\psi(\beta)\rangle_{\text{C1}}=0$.

As shown below, the canonical purification just discussed
is not the most efficient for a DMRG implementation. For this reason, in the rest of 
the section we focus on the canonical scheme of type C2: the 
maximally entangled infinite temperature state 
\begin{equation}\label{MaxEntangledCHub2}
|\psi\rangle_{\text{C2}} = P_{\binom{N_{\text{ph.}}=N}{S^{z}_{\text{ph.}}=0}}\Bigg[\bigotimes_{i=0}^{L-1} 
\big(|0,0\rangle+|\uparrow\downarrow,\uparrow\downarrow\rangle+
\sum_{\sigma=\uparrow,\downarrow}|\sigma \bar\sigma\rangle\big)\Bigg],
\end{equation}
where the projector operator is such that the total number of electrons
\emph{and} the total $z$-component of the physical (ancilla) chain are
conserved, 
\begin{gather}
\begin{aligned}
N_{\text{ph.}}|\psi_{\beta=0}\rangle_{\text{C2}}=&N_{\text{an.}}|\psi_{\beta=0}
\rangle_{\text{C2}}=N|\psi_{\beta=0}\rangle_{\text{C2}},\\
S^{z}_{\text{ph.}}|\psi_{\beta=0}\rangle_{\text{C2}}=&S^{z}_{\text{an.}}
|\psi_{\beta=0}\rangle_{\text{C2}}=\emptyset.
\end{aligned}
\end{gather} 
With a procedure similar to that outlined in section \ref{subsubsec:tj}, it is possible to show 
(see appendix~\ref{App:AppendixC}) that an
``entangler'' Hamiltonian giving Eq.~(\ref{MaxEntangledCHub2}) as ground state is 
\begin{equation}\label{hami_canonical_HubC2}
H^{\text{Hubbard}}_{\text{C2}} = -\sum_{\substack{i\neq j\\\sigma=\uparrow,\downarrow}}
 \Lambda^\dagger_{\sigma,i}\Lambda_{\sigma,j}+ {\rm h.~c.},
\end{equation}
where $\Lambda_{i,\sigma}=c_{i,\bar\sigma}c_{i,\sigma}P^{\sigma}_{i}$, where $P^{\sigma}_{i}$
is given below Eq.~(\ref{EntanglerHubGC}).

\begin{figure}
\centering
\includegraphics[scale=0.9]{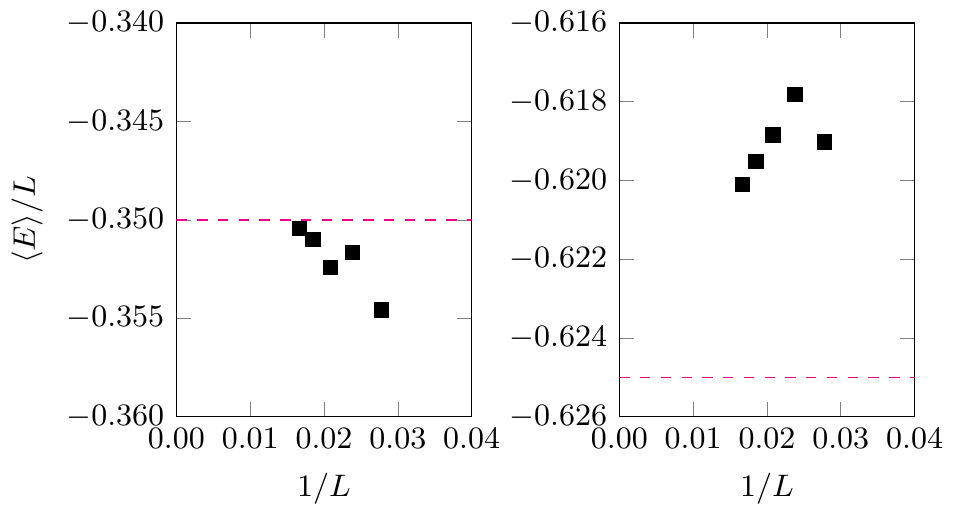}
\caption{(Color online) Intensive energies for the Hubbard model, $U=10$, and 
density $2/3$, as a function of the system size. Squares indicate data obtained with 
DMRG using a maximum of $m=1000$ and a truncation error of $10^{-6}$. 
The dashed (magenta) lines indicate the energy calculated with thermodynamic Bethe ansatz. 
Left panel shows $T=1.25$, right panel $T=0.25$.} \label{fig4a}
\end{figure}
\begin{figure}
\centering
\includegraphics{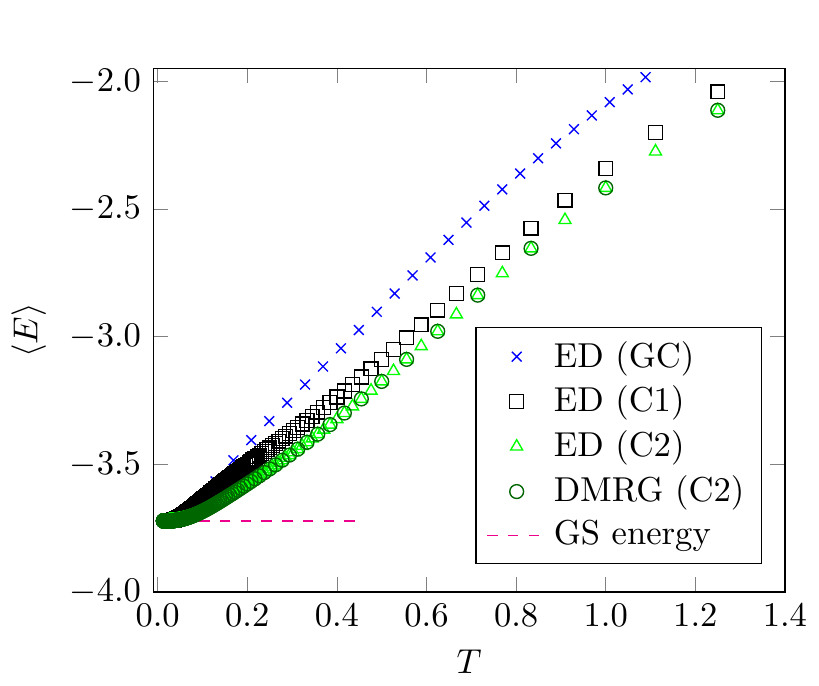}
\caption{(Color online) Intensive energies for the Hubbard model, $U=10$, and 
density $2/3$, as a function of the temperature for a $L=6$ chain. 
Shown are results within the grand canonical, calculated using ED, blue crosses; 
canonical $\text{C1}$, calculated using ED,  black squares; canonical C2 calculated using ED, 
green triangles; and canonical C2 obtained with DMRG, dark green empty circles. 
The ground state solution is also indicated with a magenta dashed line.} \label{fig4}
\end{figure}
For an $L$ site chain, the subspace containing 
the maximally entangled state in the canonical approach $\text{C1}$ has dimension 
$D_{\text{C1}}=\sum_{N_{\uparrow}=0}^{N} C^{L}_{N_{\uparrow}}C^{L}_{N_{\uparrow}}$, 
which is larger than the analogous subspace dimension $D_{\text{C2}}=C^{L}_{N/2}C^{L}_{N/2}$ 
of the case C2; both cases have $N_{\uparrow}=N_{\downarrow}=N/2$.
(Notice the difference with 
the corresponding expressions provided for the $t-J$ model in the previous section.) 
If one imposes the conservation of $N_{\text{ph.}}$ \emph{and} $S^{z}_{\text{ph.}}$ 
during the finite temperature evolution in the enlarged Hilbert
space, one can thus reduce (in our typical runs by a factor of two or more) 
the computational time needed for obtaining the 
thermodynamic properties. We therefore recommend that the canonical approach C2 
be the preferred purification 
scheme when finite temperature static and even dynamic properties are calculated with DMRG. 

Figure~\ref{fig4a} shows the intensive energies as a function of the system size for the 
canonical approach C2 for $U=10$ ($t=1$ is assumed as unit of energy) and filling $2/3$. 
The left panel shows data at $T=1.25$, while the right panel at $T=0.25$. Using $m=1000$ 
and imposing a truncation error not bigger than $10^{-6}$ in the DMRG runs, 
the results get closer to the thermodynamic Bethe ansatz results as the system 
size increases. In this paper, we have considered chains of length up to $L=60$.

Fig.~\ref{fig4} shows a comparison between the intensive energies obtained in the canonical 
and grand canonical approaches for a Hubbard chain with $L=6$ sites, $N=4$ electrons, and $U=10$
as a function of temperature. Because 
the system size is small, the numerical results differ in the three approaches, even though they converge 
to the exact 
ground state solution in the limit of $\beta\rightarrow\infty$. Remarkably, even for $L=6$, the results of the
 C1 and C2 approaches differ by less than $5\%$ on the entire $\beta$ interval studied. 
 
For the canonical purification scheme C2, Fig.~\ref{fig4} also compares 
DMRG results against those obtained with ED. 
The agreement is close to the numerical precision. Finally, the computationally expensive results 
with the grand canonical approach are also shown. These results are calculated by 
performing a procedure similar to a standard Maxwell construction.

\begin{figure}
\centering
\includegraphics{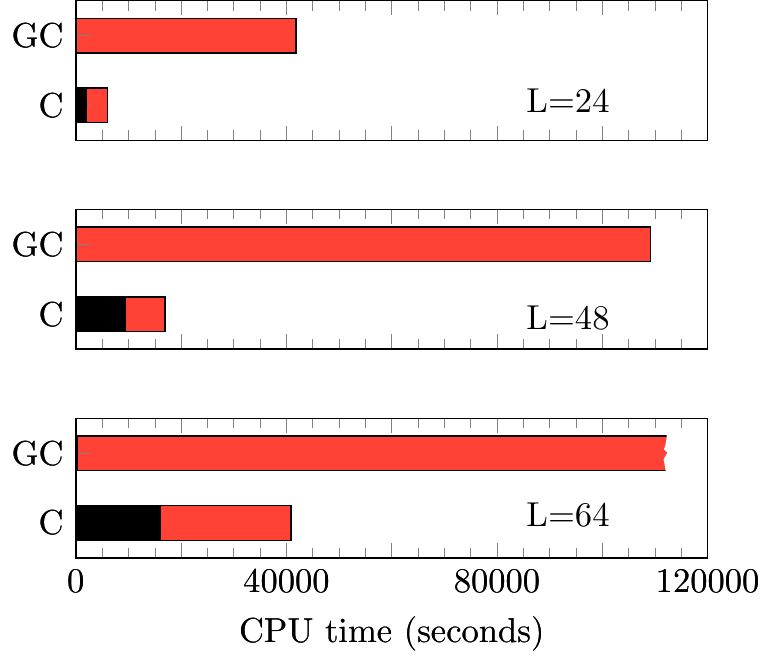}
\caption{(Color online) CPU time of typical runs in the grand canonical 
and grand canonical approaches for the Heisenberg model at $\beta=0.8$, $J=-1$, 
and different system sizes. The black bars indicate the CPU time needed to solve the entangler Hamiltonian
 (for the GC see Eq.~(\ref{HeisenbergGC}), for C see Eq.(~\ref{HeisenbergC})), 
 while the red bars indicate the same quantity for the temperature evolution.
 In the top panel ($L=24$) and middle panel ($L=48$) $m=400$ DMRG states are used for the temperature evolution, 
 while in the bottom panel ($L=64$) $m=600$. The CPU time for GC with L=64
has been truncated (the actual value is 595,000), because it is larger
than the maximum CPU time shown in this scale.} \label{fig5}
\end{figure}

\section{Long range interactions in canonical entanglers} \label{sec:entanglement}

In this section, we discuss the properties of the entangler Hamiltonians in the canonical purification scheme. As already mentioned in the previous sections, the entangler Hamiltonians 
have long range interactions, with connections between sites at all possible distances.   
The resulting entanglement growth makes it difficult to compute the ground state of such Hamiltonians with the DMRG. 
Therefore, one would think that the grand canonical
approach, where the \emph{local} entangler Hamiltonians are used, would be more efficient 
for the calculation of the thermodynamic quantities. 
But the canonical purification scheme is computationally much more efficient.

\begin{figure}
\centering
\includegraphics[]{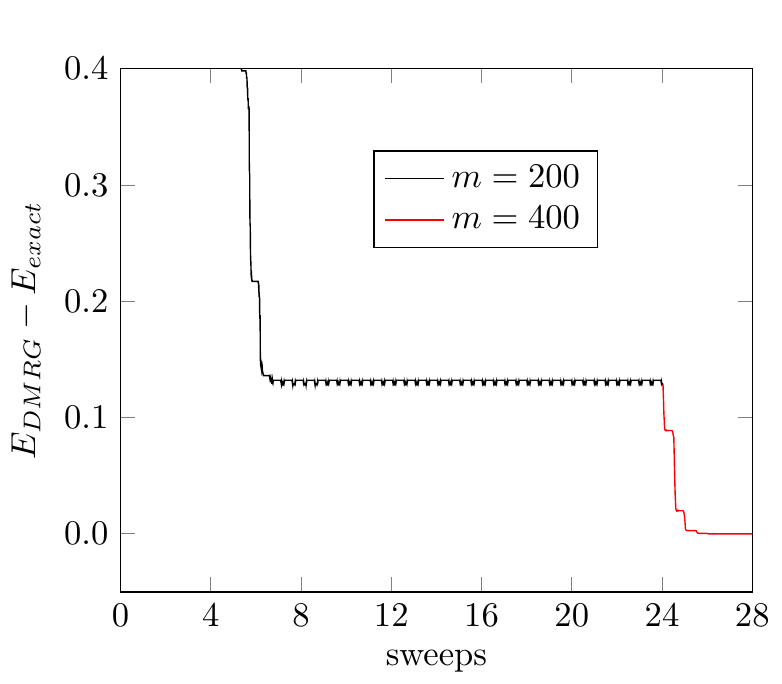}
\caption{(Color online) Difference between the ground state energy of the entangler
 Hamiltonian Eq.~(\ref{hami_canonical_HubC2}) for the Hubbard model at density $2/3$
and the exact ground state energy 
as a function of the number of DMRG sweeps. 
A chain with $L=60$ sites is considered. Increasing the number of DMRG states kept from $m=200$ to $m=400$
 allows the results to converge to the exact ground state energy.} \label{fig6}
\end{figure}

Fig.~\ref{fig5} explains why. It shows the CPU times 
of the two purification schemes for the Heisenberg model. 
CPU times are divided in two parts: the calculation of the ground state of the entangler Hamiltonian (black bars), 
and temperature evolution up to the desired inverse temperature $\beta=0.8$ (gray (red) bars), 
starting from the ground state of the entangler Hamiltonian.
The computational time
 needed to calculate the ground state of the entangler in the grand canonical approach is too small to be visible
 in the figure.
 The CPU time 
 needed for the canonical entangler Hamiltonian is larger and visible as black bars. 
The \emph{apparent} disadvantage of the canonical purification versus the grand canonical scheme seems to
increase if large system sizes are considered. For $L=64$ sites and $m=200$, the CPU time
 needed in the canonical scheme for the ground state calculation is about three orders of magnitude longer than
that needed in the grand canonical approach. 

As indicated
by the gray (red) bars in Fig.~\ref{fig5},
the CPU time of the temperature evolution of the canonical 
purification is remarkably smaller than that of the grand canonical, because
the temperature
 evolution takes place in a Hilbert space of reduced dimensions.
(This results from the conservation of $S^z_{ph}$, as stated in sec.~\ref{sec:heisenberg}.)
Given that the total CPU time is equal to the sum of the CPU time needed for the entangler plus the
CPU time needed for the temperature evolution, we can conclude that the canonical purification scheme is 
faster than the grand canonical
  approach---by one order of magnitude or more.

Yet it is important to make sure that ground states of canonical entanglers are converged, given the aforementioned
difficulties due to the long range interactions. 
To that aim, Fig.~\ref{fig6} is an example showing that the canonical entanglers used 
throughout this paper are converged.
Fig.~\ref{fig6} shows the ground state energy of the canonical entangler C2 for the Hubbard model as a function of
  the number of DMRG sweeps, for a chain of $L=60$ sites with $U=10$ and density $2/3$.
Several sweeps and two values of $m$ were used to converge to the exact ground state. 

We end this section by providing
the ground state energies of the entangler Hamiltonians.
Let $L$ be the number of physical sites of the system.
As shown in Appendix~\ref{App:AppendixB}, 
the ground state energy of the Heisenberg entangler Eq.~\ref{HeisenbergC} is -$L^2/4$.
Let $N$ be the number of electrons in the physical system.
Section II.B.2 shows that the ground state energy of the t-J entangler Eq.~(\ref{EntanglertJC1}) is $-N(L-N)$, which is equal to that of Eq.~(\ref{EntanglertJC2}).
Appendix~\ref{App:AppendixC} proves that the ground state energy of the Hubbard entangler Eq.~(\ref{hami_canonical_HubC2}) is $-N(L-N/2)$.

\section{Summary and Conclusions}\label{sec:conclusions}
\setlength{\tabcolsep}{4pt}
\begin{table}[b]
\begin{tabular}{lllll} \toprule
$N_{\rm e}^{\rm ph.}$ &  $S_{z}^{\rm ph.}$ & Heisenberg & $t$-$J$ Model & Hubbard\\ \midrule
No & No &  & GC, (\ref{MaxEntangledtJ}), (\ref{EntanglertJGC}) & GC, (\ref{MaxEntangledHub}),
 (\ref{EntanglerHubGC})\\
Yes & No & GC, (\ref{MaxEntangled}), (\ref{HeisenbergGC})   & C1, (\ref{MaxEntangledtJ}), (\ref{EntanglertJC1})
 & C1, (\ref{MaxEntangledC1Hub})\\
Yes & Yes & C, (\ref{MaxEntangledC}), (\ref{HeisenbergC})      & C2, (\ref{MaxEntangledCtJ1}),
 (\ref{EntanglertJC2}) & C2, (\ref{MaxEntangledCHub2}), (\ref{hami_canonical_HubC2})\\
\end{tabular}
\caption{\label{tbl:grandcanonical}
For each row, the conservation of the number of particles $N_{\rm e}^{\rm ph}$ or 
number of particles $N_{\rm e}^{\rm ph}$ and the $z$ component of the spin $S_{z}^{\rm ph}$ 
of the physical system is indicated by 
``Yes'' or ``No.'' For each model, the labels grand canonical (GC), 
canonical type 1 (C1)
or type 2 (C2) are defined. The numbers refer to equations in this paper: 
the first number
is the equation defining the infinite temperature state, the second number 
(if present) is the equation defining the corresponding \emph{entangler} Hamiltonian. 
For the Heisenberg model the $N_{\rm e}^{\rm ph}$ column should be ignored.}
\end{table}

In summary, we have improved the efficiency of the ancilla method 
for finite temperature DMRG by employing the \emph{inherent} symmetries
of the physical system in consideration. We have designed different entangler Hamiltonians 
to obtain infinite temperature states to use as
starting states for the temperature evolution. Table~\ref{tbl:grandcanonical} reports the
purification schemes adopted for each model considered, 
and serves as an index to the equations obtained. 
The supplemental
material can be found at \url{https://web.ornl.gov/~gz1/papers/55/}. 
It provides a pointer to the full open source code, input decks 
and additional computational details.

As a main result, we have derived entangler Hamiltonians for the canonical purification scheme of spin chains,
the $t-J$ model and the Hubbard model.
The present work
codifies an efficient ancilla method for spin chains and fermionic systems, because
(i) canonical purification is the most
efficient for obtaining the thermodynamic properties of the physical system in consideration, and
(ii) entangler  Hamiltonians are needed to obtain workable representations of infinite temperature states.
Moreover, the efficiency brought about by the use of symmetries does not compromise accuracy: 
the ancilla method is as accurate as originally proposed.

Due to their long range interactions, canonical entanglers appear computationally costlier than grand canonical ones.
But grand canonical entanglers have to be simulated on larger Hilbert spaces. 
\emph{Overall, canonical entanglers turn out to be computationally much more efficient,} as was shown in section~\ref{sec:entanglement}. 

In the Hubbard model away from half-filling, we have verified that the efficiency 
gain of the canonical scheme overcomes the extra cost due to the need for
more states in the temperature evolution for small and medium-size systems (up to $L=60$). 
For very large system sizes and at half filling, the grand 
canonical approach remains an alternative method for the calculation 
of the thermodynamic properties~\cite{Karrasch2014}. 
Yet the grand canonical scheme away from half-filling
requires computational runs to adjust the chemical potential 
at every temperature!

The purification schemes 
proposed are general and can be applied to both more complex one dimensional models and
to geometries beyond chains. 
They apply to more general models because the Hamiltonian of the system comes into play only in the temperature 
 evolution.
They apply to more general geometries because 
 the ancillas can be thought of as an extra orbital for each
physical one.

Concerning the Hubbard and t-J models, we believe that the canonical approach C2 should be 
the preferred purification scheme when finite temperature static and even dynamic properties 
are calculated with DMRG. For instance, in ref.~\onlinecite{Karrasch2013}, the entanglement
growth characterizing the purification scheme has been reduced by time evolving 
the auxiliary degrees of freedom backward in time, when a combination of finite 
temperature and time dependent DMRG is needed for the calculation of spectral properties.
 We think that a natural step would be to combine the suggested improvement with the canonical 
purification proposed in this paper.

\begin{acknowledgments}
We would like to thank T. Barthel,  E. Dagotto, A. Feiguin, C. Karrasch,  S. Manmana,
U. Schollw\"oeck, and M. Stoudenmire for
helpful discussions and suggestions.
This work was conducted at the Center for Nanophase Materials Sciences, sponsored 
by the Scientific User Facilities Division, Basic Energy Sciences, Department 
of Energy (DOE) (USA), under contract with UT-Battelle. We acknowledge support by the early 
career research program Department of Scientific User Facilities, Office of Science, Basic Energy Sciences, U.S. DOE.
\end{acknowledgments}

\appendix
\section{Canonical vs. grand canonical ensemble} \label{App:AppendixA}
The purpose of this appendix is three-fold. First, to show that, for finite systems at finite temperature, 
the canonical and grand canonical
ensembles do not, in general, yield the same averages. Second, to show that, at $T=0$, the   canonical and
 grand canonical
ensembles agree for any Hamiltonian. And third, that they also agree \emph{in the thermodynamic limit} for any
 Hamiltonian.
Even though these results are known, we prove them below for completeness.

\begin{myprop}
At finite temperature, there exists at least one finite dimensional Hamiltonian 
with canonical average energy different than its grand canonical average energy.
\end{myprop}
\begin{proof}
We construct such a Hamiltonian by considering 2 levels, 
and a one-site Hilbert space including the states,
empty, level 1 occupied only, level 2 occupied only, and both levels occupied.
Let
\begin{equation}
H=\epsilon(\hat n_{1}+\hat n_2)+V\hat n_{1}\hat n_{2},
\end{equation}
where $\hat{n}_l$ acting on the basis states multiplies it by 
the level $l$ occupation of that state.
Let $\epsilon>0$ and $V>0$.
The average energy in the canonical ensemble 
with exactly $N=1$ particle in the system is $\langle E\rangle_{C}=\epsilon$.
The average energy in the grand-canonical ensemble 
 with density $\langle N\rangle=1$ is given by
$\langle E\rangle_{GC}=\epsilon + \frac{V/2}{1+{\rm e}^{\beta V/2}},$
which is greater than $\epsilon$ at finite temperature.
\end{proof}

\begin{myprop}
For any finite dimensional Hamiltonian with convex energies, at zero temperature, 
the average energy in the canonical ensemble is the same as the one in the grand canonical
ensemble.
\end{myprop}
\begin{proof}
Let us consider a system with $M$ sites and target density $N_T$. The number $M$
includes sites, orbitals and spin, such that the maximum number of electrons that
the system can hold is $M$.
Let the eigenvalues of the Hamiltonian be $E_{n'}^N$, where the index $n'$ runs
only over states of Fock sector $\mathcal{F}_N$ of $N$ particles.
Let $E_{\rm min}^N=
\min_{n'\in\mathcal{F}_N}E_{n'}^N$. 
At $T=0$, the average energy in the canonical ensemble with density $N_T$ is $E_{\rm min}^{N_T}$.  
In the grand canonical (GC) ensemble we impose $\langle N\rangle_{{\rm GC},\,T=0}=N_T$,
and thus the chemical potential $\mu$ 
is obtained by imposing that the $N$ that minimizes
 $F_N\equiv E_{\rm min}^N-\mu N$ be $N_T$. Because the energies are convex, such a $\mu$ is unique. 
 For this $\mu$, we obtain $\langle E\rangle_{{\rm GC},\,T=0}=E^{N_T}_{\rm min}$,
which coincides with the canonical result.
\end{proof}

\begin{myprop}
For any Hamiltonian with convex energies and an extensive canonical partition function, at finite temperature,
and in the thermodynamic limit, the average energy in the canonical ensemble is the same as the one in the
grand canonical ensemble.
\end{myprop}
\begin{proof}
The proof is given in the supplemental material at \url{https://web.ornl.gov/~gz1/papers/55/}.
\end{proof}

\section{Canonical entangler for the Heisenberg model} \label{App:AppendixB}
We now prove that the ``canonical'' entangler Hamiltonian for the Heisenberg model, Eq.~(\ref{HeisenbergC}),
has the property that its ground-state is Eq.~(\ref{MaxEntangledC}).
We begin by observing that the Hamiltonian (\ref{HeisenbergC}) conserves the 
$z$-component of the total spin in the physical and ancilla chains \emph{separately}, 
and assume that the combined physical plus ancilla system has $S_z=0$, so that 
the number of up and down spins are equal in the
combined system. Let
\begin{gather}
\begin{aligned}
\mathcal{T}_G=&\{|\phi\rangle; \text{basis } |\phi\rangle\text{ with all physical and ancilla}\\
&\text{ sites correctly paired}\}\\
\mathcal{T}_B=&\{|\phi\rangle; \text{basis } |\phi\rangle\text{ with at least one physical site}\\
&\text{ with ancilla incorrectly paired}\},
\end{aligned}
\end{gather}
where $\mathcal{T}_G$ is the set of states in which each site is correctly paired with
ancilla according to the scheme $\uparrow\rightarrow\downarrow$,
$\downarrow\rightarrow\uparrow$,
and $\mathcal{T}_B$ the set where at least one site is incorrectly paired.

It is easy to prove (i) that 
$\langle \phi'|H^{spin}_{\text{C}}|\phi\rangle=0$ if $|\phi'\rangle\in\mathcal{T}_G$
and $|\phi\rangle\in\mathcal{T}_B$,
so that the Hamiltonian matrix $H^{spin}_{\text{C}}$ blocks into at least
two blocks: states in $\mathcal{T}_G$ and states in $\mathcal{T}_B$.
To this aim, it is sufficient to think of spins as mapped to hard-core bosons, 
$|\uparrow\rangle=|1\rangle$ 
$|\downarrow\rangle=|0\rangle$. One can imagine Eq.~(\ref{HeisenbergC}) as a 
tight-binding Hamiltonian of ``singlets''
of physical sites and ancillas built along the rungs, singlets that
are exchanging positions of particles and holes. Because Eq.~(\ref{HeisenbergC}) conserves
 the number of
 bosons in the physical and ancilla chain separately,  it cannot, 
therefore, connect states in $\mathcal{T}_G$ (number of bosons and holes is equal to 
$L/2$ in the physical
chain)
with those in $\mathcal{T}_B$.
In fact, in the subspace
$\mathcal{T}_B$ at least one site of the chain has, in the hard-core 
bosonic representation, occupation $|0,0\rangle$ ($|\downarrow,\downarrow\rangle$), 
$|1,1\rangle$ ($|\uparrow,\uparrow\rangle$), 
bosons in sites like these cannot ``move''
by the action of  Eq.~(\ref{HeisenbergC}). As in the other cases, the subspace
$\mathcal{T}_B$ then blocks into even smaller subspaces, all of size smaller 
than the set $\mathcal{T}_G$. Within these subspaces we have hopping-like
matrices yielding lowest energies which are larger than those in the 
set $\mathcal{T}_G$.

To prove (iii), let us call $H$ the block of matrix Eq.~(\ref{HeisenbergC}) in the
$\mathcal{T}_G$ subspace.
The rank of $H$ is $C^L_{L/2}$ and its matrix elements are either 0 or -1. 
Each row has exactly $L^2/4$ non-zero entries. So does each column. Then its lowest eigenvalue is
$-L^2/4$, with eigenstate $(1,1,\cdots,1)$, that is,  
Eq.~(\ref{MaxEntangledC}).\\
\section{Canonical entangler for the Hubbard model} \label{App:AppendixC}
We now prove that the ``canonical'' entangler Hamiltonian of type C2 for the Hubbard model,
Eq.~(\ref{hami_canonical_HubC2}),
has the property that its ground-state is Eq.~(\ref{MaxEntangledCHub2}).
We begin by observing that the Hamiltonian (\ref{hami_canonical_HubC2}) conserves the 
number of electrons and $z$-component of the total spin in the physical and ancilla chains \emph{separately}, 
and assume that the combined physical plus ancilla system has $S_z=0$, so that 
the number of up and down electrons are equal in the
combined system. As done in Eq.~\ref{goodbad}, let
\begin{gather}
\begin{aligned}
\mathcal{W}_G=&\{|\phi\rangle; \text{basis } |\phi\rangle\text{ with all physical and ancilla}\\
&\text{ sites correctly paired}\}\\
\mathcal{W}_B=&\{|\phi\rangle; \text{basis } |\phi\rangle\text{ with at least one physical site}\\
&\text{ with ancilla incorrectly paired}\},
\end{aligned}
\end{gather}
where $\mathcal{W}_G$ is the set of states in which each site is correctly paired with
ancilla according to the scheme $0\rightarrow0$, $\uparrow\rightarrow\downarrow$,
$\downarrow\rightarrow\uparrow$, $\uparrow\downarrow\rightarrow\uparrow\downarrow$, 
and $\mathcal{W}_B$ the set where at least one site is incorrectly paired.

Even in this case, it is easy to prove (i) that 
$\langle \phi'|H^{Hub}_{\text{C2}}|\phi\rangle=0$ if $|\phi'\rangle\in\mathcal{W}_G$
and $|\phi\rangle\in\mathcal{W}_B$,
so that the Hamiltonian matrix $H^{Hub}_{\text{C2}}$ blocks into at least
two blocks: states in $\mathcal{W}_G$ and states in $\mathcal{W}_B$.
Indeed, it is sufficient to imagine Eq.~(\ref{hami_canonical_HubC2}) as a 
tight-binding Hamiltonian of generalized ``singlets'
of physical sites and ancillas built along the rungs

\begin{equation}
|\psi_{\rm{rung~singlet}}\rangle =|0,0\rangle+|\uparrow\downarrow,\uparrow\downarrow\rangle + 
|\uparrow,\downarrow\rangle+ |\downarrow,\uparrow\rangle.
\end{equation}

Because Eq.~(\ref{hami_canonical_HubC2}) conserves the number of electrons in the physical and ancilla
chain separately,  it cannot, 
therefore, connect states in $\mathcal{W}_G$
with those in $\mathcal{W}_B$. In fact, in the subspace
$\mathcal{W}_B$ at least one site of the chain has occupation $|\sigma,0\rangle$, $|0,\sigma\rangle$, 
$|\uparrow\downarrow,\sigma\rangle$, $|\sigma,\uparrow\downarrow\rangle$ or $|\sigma,\sigma\rangle$.
Electrons in sites like these cannot ``move''
by the action of  Eq.~(\ref{hami_canonical_HubC2}). As in the other cases, the subspace
$\mathcal{W}_B$ then blocks into even smaller subspaces, all of size smaller 
than the set $\mathcal{W}_G$. Within these subspaces we have hopping-like
matrices yielding lowest energies which are larger than those in the 
set $\mathcal{W}_G$.

To prove (iii), let us call $H$ the block of matrix Eq.~(\ref{hami_canonical_HubC2}) in the
$\mathcal{S}_G$ subspace.
It can be written as a direct product 
$H=H_\uparrow\otimes H_\downarrow$. $H_\uparrow=H_\downarrow$, each has rank $C^L_{N/2}$, and matrix elements either 0 or -1. 
Each row has exactly $N(L-N/2)/2$ non-zero entries. So does each column. Then the lowest eigenvalue of $H$ is
$-N(L-N/2)$, with eigenstate $(1,1,\cdots,1)$, that is,  
Eq.~(\ref{MaxEntangledCHub2}).
\bibliography{biblio}

\end{document}